\documentclass[11pt,a4paper]{article}
\usepackage[english]{babel}
\usepackage{amsmath,amssymb,amsthm,epsfig,color,graphicx}
\usepackage[latin1]{inputenc}

\usepackage{tikz, tkz-euclide}
\usetikzlibrary{calc}
\usetikzlibrary{decorations.markings}

\usepackage{bbm}
  \newtheorem*{theorem*}        {Theorem}
	\newtheorem*{conjecture*}   {Conjecture}
  \newtheorem{theorem}           {Theorem}
  \newtheorem{lemma}              {Lemma}
  \newtheorem*{lemma*}          {Lemma}
  \newtheorem{definition}         {Definition}
  \newtheorem{corollary}          {Corollary}
  \newtheorem{proposition}      {Proposition}

\usepackage{xcolor}

\definecolor{Red}{cmyk}{0,1,1,0}

\definecolor{Blue}{cmyk}{1,1,0,0}

\newcommand{\ba}{\begin{array}}
\newcommand{\ea}{\end{array}}
\newcommand{\be}{\begin{equation}}
\newcommand{\ee}{\end{equation}}
\newcommand{\ben}{\begin{enumerate}}
\newcommand{\een}{\end{enumerate}}

  \def\E{\mathop{\textrm{\rm E}}\nolimits}                  
  \def\d{\mathop{\textrm{\rm d}}\nolimits}                  

\let\m=\mu
\let\n=\nu

\let\L=\Lambda

\newcommand{\La}{\Lambda}
\newcommand{\R}{\mathbb{R}}

\newcommand{\Z}{\mathbb{Z}}
\newcommand{\N}{\mathbb{N}}

\begin{document}

\title{ Entropic repulsion and lack of the $g$-measure property for Dyson models}
\author{
Rodrigo Bissacot\\
\footnotesize{\texttt{rodrigo.bissacot@gmail.com}}\\
\footnotesize{Institute of Mathematics and Statistics - IME USP - University of S\~ao Paulo, Brazil}
\\[0.3cm]
Eric Ossami Endo\\\
\footnotesize{\texttt{eric@ime.usp.br}}\\
\footnotesize{Institute of Mathematics and Statistics - IME USP - University of S\~ao Paulo, Brazil}\\
\footnotesize{Johann Bernoulli Institute for Mathematics and Computer Science - University of Groningen, Netherlands}\\
\\
Aernout C. D. van Enter\\
\footnotesize{\texttt{a.c.d.van.enter@rug.nl}}\\
\footnotesize{Department of Mathematics}\\
\footnotesize{Johann Bernoulli Institute for Mathematics and Computer Science - University of Groningen, Netherlands}\\
\\
Arnaud Le Ny\\
\footnotesize{\texttt{arnaud.le-ny@u-pec.fr}}\\
\footnotesize{Laboratoire de Math\'ematiques et d'Analyse Appliqu\'ees}\\
\footnotesize{LAMA UMR CNRS 8050 -- Universit\'e Paris-Est (UPEC), Cr\'eteil, France}\\
}
\maketitle

\begin{abstract}
 We consider Dyson models, Ising models with  slow polynomial decay, at low temperature and show that its Gibbs measures deep in the phase transition region are not $g$-measures. The main ingredient in the proof is the occurrence of an entropic repulsion effect,  which follows from the  mesoscopic stability of a (single-point) interface for these long-range models in the phase transition region.
\end{abstract}



\section{Introduction}
 
Dyson models, long-range Ising models with  ferromagnetic, polynomially decaying, pair interactions, have been studied for a considerable time. After Dyson \cite{Dys1, Dys2} proved the existence of a phase transition, confirming a conjecture due to Kac and Thompson \cite{Kac}, various alternative proofs and further properties have been derived. One recent low-temperature result which we will find particularly useful is the existence  of phase separation, properly defined, with an ``interface point", which is  to some extent stable  under infinite-volume limits with appropriate mixed boundary conditions similar to Dobrushin boundary conditions introduced in higher dimensions. Indeed, in \cite{CMPR} it was shown that a Dyson model in a finite interval of length L, with $-$-boundary conditions on the left and $+$-boundary conditions on the right, has an interface of ``mesoscopic size" for decay parameter values
\footnote{Our results will be valid only  for $\alpha$ satisfying the lower bound $\alpha > \alpha_+$ -- already present in \cite{CFMP, CMPR, COP}. In contrast  to the upper bound $ \alpha < \alpha=2$, we believe this lower bound is technical only, as we shall see.} 
$\alpha_+ < \alpha < 2$, once the temperature is low enough (but non-zero). This means that with overwhelming probability its location is in the middle of the interval, up to a Gaussian correction which grows sublinearly with $L$.

In this paper we notice that this  interface result implies in a fairly straightforward manner  that a form of entropic repulsion occurs, in the sense that a large interval of minuses inserted in the $+$-phase has two moderately large intervals around it\footnote{ They are the "wet" regions,  while frozen interval is  a hard wall in  a "complete wetting" situation.} in which the system will be in the  $-$-phase. We use this observation to show that the low-temperature Gibbs measures of the Dyson model are not $g$-measures: their conditional probabilities w.r.t. the past are not necessarily continuous functions of this past. It was shown before that there exist $g$-measures which are not Gibbs measures \cite{FGM}; our result  answers  a question raised in \cite{FM1} and shows that neither class of measures contains the other one.  Although the question had been posed before, it seems to be the case that  there were no precise conjectures whether these Dyson Gibbs measures actually were $g$-measures or not. We thus elucidate  a somewhat unclear situation, about the connection between  two similar-looking notions,  originating in  two different fields of research (namely Mathematical Statistical Mechanics and Dynamical Systems). 

{\bf Warning:} The case $\alpha =2$ is somewhat different; as the fluctuations in the location of the interface are macroscopic, rather than mesoscopic \cite{CMPR}, our arguments do not fully work in that case. We also note that the proof(s) and even the properties of the phase transition for this borderline case had already required a special treatment before. The model gives rise to a more complex situation in which an intermediate phase arises \cite{IN}, and also a discontinuity  of the critical magnetization occurs \cite{ACCN}.

\section{Definitions, notation and main result}

\subsection{Dyson models}
We consider Ising spins for configurations $\omega \in \{-1,+1\}^\Z$ which have a ferromagnetic long-range pair interaction, with decay parameter $1 < \alpha < 2$, of the (formal) form:
$$
H(\omega) = -\sum_{i,j \in \Z} |i-j|^{-\alpha}\omega_i \omega_j.
$$
It has been known since \cite{Dys1} (and later \cite{Dys2, frsP} for $\alpha=2$) that these models at low temperature display a phase transition. There  are a non-zero spontaneous magnetisation $m=m(\alpha,\beta) >0$ and two (extremal) Gibbs measures $\mu^{+}$ and~$\mu^{-}$ obtainable by $+$- or $-$-boundary conditions, such that $\mu^{\pm} = \pm m$ \cite{Joh,FILS,CFMP,ACCN}. It is also known that  there are no non-translation-invariant extremal Gibbs measures (\cite{Geo}, Theorem 9.5). This is usually interpreted as  the absence of interface on a microscopic scale. 
However, at mesoscopic scales,  in between the microscopic and the macroscopic scales, interfaces still  may be identified \cite{CMPR}. 

To be specific, but without loss of generality, we will consider the plus measure $ \m^{+}$, obtained e.g. by taking the weak limit with the homogeneous $+$-boundary conditions.  A similar analysis could be performed for  the minus measure $ \m^{-}$, similarly obtained  by taking the limit with the homogeneous $-$-boundary conditions. In the regime we consider, those are the only two extremal Gibbs measures.

We will also consider  Dobrushin boundary conditions, where the spin of the sites outside the interval is minus to the left and plus to the right, i.e. $\omega_i=+1$ if $i\ge 0$ and $\omega_i=-1$ if $i<0$. It is known that in that case, when we consider the box $\Lambda_L=[-L,L]$, there are $2L+2$ ground states. This differs from the $+$- and$-$- boundary conditions, for which there is only one ground state. The (mostly finite-volume) Gibbs measures obtainable by Dobrushin boundary conditions will be denoted by  ``$\mu^{-+}$''.

\subsection{Gibbs measures and $g$-measures}
\subsubsection{General definitions and main result}

 In Mathematical Statistical Mechanics, in the framework\footnote{The so-called DLR approach as described also for example in \cite{EFS,FV,Geo, Isr1, Preston}.} initiated by Dobrushin \cite{Dob}, and Lanford and Ruelle \cite{LaR}, Gibbs measures at infinite volume are probability measures, defined by conditional probabilities\footnote{When not described more precisely, conditional probabilities always are defined only almost-surely.}, conditioned on (sets of) configurations on  the outside of finite sets $\Lambda$.  On the exterior, that is the complement of $\Lambda$,  boundary conditions are frozen to provide  within the finite volume the corresponding Boltzmann-Gibbs weights in terms of Hamiltonians, in the sense that one has for all configurations $\omega^1, \omega^2$ and ($\mu$-a.e.) boundary conditions $b$ $\in \{-1,+1\}^{\Z}$, 
 
$$
\frac{\mu(\omega^{1}_{\Lambda} | b_{\Lambda^c})}{\mu(\omega^{2}_{\Lambda}|b_{\Lambda^c})} = e^{-\beta [H(\omega^{1}_{\Lambda}b_{\Lambda^c}) - H(\omega^{2}_{\Lambda} b_{\Lambda^c})]}.
$$

As we consider Ising spins, which are discrete as well as compact, continuity 
(in the product topology) coincides with quasilocality. Quasilocal functions  are uniform  limits of local (cylinder) functions and  quasilocal measures are those measures whose conditional probabilities w.r.t. the outside of finite sets always admit a regular version that is continuous as a function of the boundary condition. Up to  a ``non-nullness'' or ``finite-energy'' condition, Gibbs measures {\em are} the quasilocal measures. See e.g. \cite{EFS, Geo, Ko, Su}. In fact, in the context of possibly non-Gibbsian renormalized Gibbs measures \cite{EFS, ELN}, the major characterisation used of the latter was precisely the lack of this quasilocality property (as well as the main drawback, preventing many standard results).

In our one-dimensional setting,  a basis of neighborhoods for a configuration $\omega$ in the configuration space $\Omega:=\{-1,+1\}^\Z$ can be chosen of the form 
$$
\mathcal{N}_L(\omega)=\left\{ \sigma \in \Omega : \sigma_{\Lambda_L}=\omega_{\Lambda_L},\; \sigma_{\Lambda_L^c} \; {\rm arbitrary} \right\},\;L \in \N,\;
$$
where  $\Lambda_L:=[-L,L]$ is the set $\{-L,-L+1,\ldots,L-1,L\}$, and $\omega_{\Lambda_L}$ the restriction of $\omega$ to the sites in $\Lambda_L$.
For any integers $N > L$, we shall also consider particular open subsets of neighborhoods 
$\mathcal{N}_{N,L}^+(\omega)$ (resp. $\mathcal{N}_{N,L}^-(\omega)$) on which the configuration is $+$ (resp. $-$) on the annulus $\Lambda_N \setminus \Lambda_L $ for $N > L$:
$$
\mathcal{N}_{N,L}^+(\omega) = \left\{ \sigma \in \mathcal{N}_L(\omega) : \sigma_{\Lambda_N \setminus \Lambda_L}  = +_{\Lambda_N \setminus \Lambda_L}\right\} \; \left({\rm resp.} \; \mathcal{N}_{N,L}^-(\omega)\right),
$$
where for $\Lambda \subset \Z$, $+_{\Lambda}$ is the configuration in $\Lambda$ in which all the spins are plus. Similarly we define the one-sided equivalent objects, such as $\mathcal{N}_{N,L}^{+,left}(\omega)$ (resp.  $\mathcal{N}_{N,L}^{-,left}(\omega)$) when the $N$ spins to the left of the interval $\Lambda_L$ are constrained to be plus (resp. minus).\\



Considering the lattice $\Z$ as a bi-infinite sequence of ``times'', it is tempting to consider measures on $\Omega$ as stochastic processes (and to transfer the Gibbs property to some Markovian-like or almost-Markovian property). This equivalence holds in particular under conditions of weak coupling, such as when a Dobrushin uniqueness condition holds, for example for long-range Dyson models at high temperature, as well as for short-range models in which the coupling between two infinite half-lines is uniformly bounded. In the latter case the equivalence holds at all temperatures.
 However, it is far from obvious if such a description  is always easily possible (see e.g. \cite{FM1,FM2,FGM}). In fact, the non-equivalence between one-sided and two-sided conditionings, which we will demonstrate in detail  later, serves as a warning to a too easy identification. Gibbs measures in dimension one are thus those measures for which there exists a family -- called a ``specification" \cite{Geo} -- of continuous probability kernels $\gamma_L$ with $L \in \N$ which prescribes its (regular) conditional probabilities jointly w.r.t. the past and future via $
\mu \big[\omega_{\Lambda_L} | \omega_{\Lambda_L^c}] = \gamma_L(\omega)$.
Or, in a more Markovian-like description, 
\begin{equation}\label{twosided}
\mu \big[\sigma_{-L}=\omega_{-L},\ldots, \sigma_{L}=\omega_{L} | \ldots \sigma_{-L-1}=\omega_{-L-1}, \sigma_{L+1}=\omega_{L+1}, \ldots ] = \gamma_L(\omega).
\end{equation}
Thanks to their quasilocality properties,  Gibbs measures are the non-null measures for which the $\gamma_L$ are continuous functions of $\omega$.  In this case, it is possible to reconstruct all the conditional probabilities (\ref{twosided}) from the single-site conditional probabilities at time $0$, given for $\mu$ a.e. $\omega$ by
$$
\gamma_0(\omega):= \E_\mu \big[\sigma_0  | \mathcal{F}_{\{0\}^c}\big](\omega) = \E_\mu \big[\sigma_0  | \mathcal{F}_{\{<0\} \cup \{>0\}} \big](\omega)
$$
or, more shortly
$
\gamma_0(\omega):= \mu \big[\sigma_0  |\mathcal{F}_{\{0\}^c} \big](\omega) = \mu \big[\sigma_0  | \mathcal{F}_{\{<0\} \cup \{>0\}}\big](\omega)
$. Here $\mathcal{F}_{\{<0\} \cup \{>0\}}=\mathcal{F}_{\{0\}^c}$ denotes the $\sigma$-algebra generated by the past and the future. We shall encounter  later  the past and future $\sigma$-algebras $\mathcal{F}_{\{<0\}}$ and $\mathcal{F}_{\{>0\}}$  generated by the projections indexed by negative and positive integers. The function $\gamma_0$ is a $\mathcal{F}_{\{0\}^c}$-measurable function and when the measure is a Gibbs measure, this function is continuous,  jointly in  past and future.

\smallskip
 In Dynamical Systems, $g$-measures are defined in a similar  way,  combining  topological and measurable notions, but the transition functions (the ``$g-$''functions) now have to be continuous functions of the past only. One  requires continuity of single-site one-sided conditional probabilities and says that $\mu$ is a $g$-measure if there exists a (past-measurable) {\em continuous}  and non-null  function $g_0$ which gives ``one-sided'' conditional probabilities, that is   non-null conditional probabilities for  events localised on the right halfline (``future"), given a boundary condition  fixed only to the left (``past"). 
\begin{definition}
A probability measure is a $g$-measure, if there is a non null continuous function $g_0 > 0$, defined on the left (``past'')  half-line configuration space, such that, for each $\omega_0\in \{-1,+1\}$ and $\mu$ a.e. $b=(b_j)_{j<0}\in \{-1,+1\}^{(-\infty,0)}$,
\be \label{g}
 \mu[\omega_0|\mathcal{F}_{<0}](b):= \E_\mu \big[\mathbbm{1}_{\sigma_0=\omega_0} | \mathcal{F}_{\{<0\}} \big](b) = g_0(b \omega_0).
\ee
\end{definition}
In this situation, the function $g:=g_0$ is called a $g$-function. For translation-invariant measures, it is extended to any site $i$ with conditional probabilities w.r.t. to the past at site $i$ given by $g_i=g$,  while in the absence of translation invariance, other functions $g_i$'s are introduced to get  $G$-measures \cite{BrDoo, BrDoo2}.
 The complete formalism -- providing {\em all} conditional probabilities w.r.t. to the past --  can be restored under  extra conditions via the notion of a ``Left Interval  Specification'' (LIS) \cite{FM1,FM2}.  We  focus here on the single-site properties that define $g$-functions and $g$-measures, in a translation-invariant context.

Note that such extensions of (one-sided) Markov properties have been studied under different names in various areas of mathematics for a long time, such as {\em Chains with infinite connections} \cite{Be}, {\em Chains of infinite order} \cite{Har}, {\em  Variable Length Markov Chains} \cite{GL}, {\em uniform martingales} \cite{Kal} etc.
For a number of papers addressing $g$-measures and  related properties, see e.g.
\cite{BHS, BFV, BK, BrDoo, BrDoo2, DF, Fr, GGT, GP, Hul1,Hul,JOP1,JOP, Ver}.
When the interactions are finite-range, $g$-measures are Markov chains. These coincide with Gibbs measures, which then are Markov fields, expressible in two-sided conditional probabilities, see e.g. \cite{Geo}, Chapter 3. In fact, this equivalence applies for a large class of interactions which satisfy a strong uniqueness condition \cite{FM1,FM2}. However, if we require only continuity of the conditional probabilities, there exist $g$-measures which are {\em not} Gibbs measures \cite{FGM}.

In general, there is not that much known in the phase transition region, where the interactions are necessarily long-range. Phase transitions in the Gibbs measure context have been known to occur since Dyson, and in the  $g$-measure context they are also known to be possible \cite{BK,BHS,DF, Fr, Hul}. Nevertheless, there seems little known about the equivalence of the Gibbs measure property and the $g$-measure property in any such general context. In higher dimension, one could interpret  the ``Local Markov Property'' as a Gibbs property and the ``Global Markov Property''  (see e.g. \cite{Foll}) to some extent as the equivalent of the $g$-measure property. It is known that there are measures having the Local, but not the Global Markov Property \cite{Gol,Isr2,vW}. Here we will show the somewhat analogous result that the Gibbs measures of the Dyson model are not $g$-measures.

Discontinuity of any candidate $g^+$ to represent a $g$-function for $\mu^+$ -- i.e. discontinuity of any possible version of a suitable chosen conditional probability -- will be a consequence of the next lemma, proved using an entropic repulsion phenomenon, which we obtain  as a fairly direct corollary of  the interface localisation result of \cite{CMPR}. To use these results of Cassandro {\em et al.},  we will require the  same technical lower bound $\alpha_+=3 - \frac{\log{3}}{\log{2}} \in ]1,2[$ as they needed. In the following lemma, $\mu^{+,\omega}_{\mathbb{Z}_+}[\cdot]$ denotes expectations under a constrained measure $\mu^{+,\omega}_{\mathbb{Z}_+}$, defined in the next section.





 \begin{lemma}\label{keylemma}
Consider the  alternating configuration $\omega_{\rm alt}=\big((\omega_{\rm alt})_i\big)_{i\in \mathbb{Z}}$ defined by $(\omega_{\rm alt})_i=(-1)^i$, and take a Dyson model with polynomial decay $\alpha_+ <\alpha<2$ at sufficiently low temperature. 
Then, there exist $L_0\ge 1$ and $\delta >0$ such that for any  $L>L_0$ there is an $N > L$, with $LN^{1-\alpha}=o(1)$, such that for every two configurations $\omega^+ \in \mathcal{N}_{N,L}^{+,left}(\omega_{\rm alt})$ and $\omega^{-} \in \mathcal{N}_{N,L}^{-,left}(\omega_{\rm alt}) $,
\be \label{keymagn}
\left| \mu^{+,\omega^+}_{\mathbb{Z}_+}[\sigma_0] -  \mu^{+,\omega^-}_{\mathbb{Z}_+}[\sigma_0] \right| > \delta.
\ee
\end{lemma}

As a corollary, we obtain our main result:

\begin{theorem}\label{thm:main}
 For $\mu$ being either the plus or the minus  phase of  a Dyson model with exponent $\alpha_+ <\alpha<2$ at sufficiently low temperature, the one-sided conditional probability  $\mu[\omega_0|\mathcal{F}_{<0}](\cdot)$ is essentially discontinuous at $\omega_{\rm alt}$. Therefore, none of the Gibbs measures $\mu$  for the Dyson model in this phase transition region\footnote{Note that we again impose  the technical restriction $\alpha_+ < \alpha <2$ on the decay parameter.} is a $g$-measure.
\end{theorem}

{\bf Remark 1:} We use the term Gibbs measure in the Statistical Mechanics sense, as defined by Dobrushin, Lanford and Ruelle \cite{Dob,LaR}. In the Dynamical Systems community, often a somewhat different notion of Gibbs measure is defined following Sinai, Ruelle and Bowen \cite{Sin,Rue,Bow}, 
 by providing uniformly bounded approximations of the measure on cylinders as  exponential Boltzmann-Gibbs weights defined via (a  slightly different notions) potentials.  In Symbolic Dynamics, yet another notion is introduced either via Perron-Frobenius operators or via variational principles  and a corresponding notion of equilibrium states. Compare e.g. \cite{BFV} with sometimes different (non-)lattices, and again different notions of potentials compared to the ones used in Mathematical Statistical Mechanics. This yields different, typically more restrictive, classes of measures, in which phase transitions are  usually excluded due to the corresponding interaction being too short-range (in statistical mechanics terms). 

{\bf Remark 2:} As discussed in \cite{FM1,FM2}, which discuss a lot of the history, the terminology ``$g$-measures'' was introduced by Keane \cite{Ke}, but the notion is older.  In those papers also the observation is made and exploited  that the $g$-measure property is a kind of one-sided Gibbs property. However, this analogy appears to work properly mostly in various uniqueness regimes, as we illustrate here.

\subsubsection{Gibbs vs $g$-measures for Dyson models in the Phase Transition region}

To be more specific, we consider configurations lying in the infinite probability space $(\Omega,\mathcal{F},\rho)=(E,\mathcal{E},\rho_0)^\Z$ where $E=\{-1,+1\}$ is equipped with the a priori product measure $\rho_0 = \frac{1}{2} \delta_{-1} + \frac{1}{2} \delta_{+1}$.  For a configuration $\omega \in \Omega$ and any $\Lambda \subset \Z$, we consider the restriction  $\omega_\La$ and the corresponding configuration spaces at volume $\Lambda$ as the product probability spaces $(\Omega_\Lambda, \mathcal{F}_\La,\rho_\La)$ defined in a standard way. To specify the two-sided conditional probabilities of our Dyson measures,  we consider the set $\mathcal{S}$ of finite subsets of $\Z$ and introduce the following, in  particular Gibbsian, specification (see e.g. \cite{Fer, FM3, Geo,Gold,Preston,Sok} for more details about specifications):
\begin{definition}
Let $\beta >0$ be the inverse temperature. We call a {\em Dyson specification} the collection of probability kernels $\gamma^D=(\gamma_\La^D)_{\La \in  \mathcal{S}}$ { from $\mathcal{F}_{\Lambda^c}$ to $\Omega_\La$ defined} by
\be \label{LRDysonSpe}
\gamma_\Lambda^D(d \omega | \tau) = \frac{1}{Z_\Lambda^\tau} \; e^{\beta \sum_{i \neq j, i \in \La, j \in \Z} \frac{1}{|i-j|^\alpha} \omega_i \omega_j} \; \rho_\Lambda \otimes \delta_{\tau_{\La^c}} (d \omega)
\ee
where the normalization $Z_\La^\tau$ is the usual partition function.
\end{definition}

This specification is {\em monotonicity-preserving} (or FKG): for all $\La \in \mathcal{S}$ and any $f$ bounded increasing,   so is  $\gamma_\La^D f$. The extremal (maximal and minimal) elements of this partial order ``$\leq$" already allow us to define the extremal elements of  $\mathcal{G}(\gamma^D)$:
\begin{proposition}\label{DyFrSp}\cite{Dys1,FP,frsP, Hul1} The weak limits
\be \label{muplusminus}
\mu^-(\cdot) := \lim_\La \gamma_\La^D (\cdot | -)\; \; {\rm and} \; \; \mu^+(\cdot) := \lim_\La \gamma_\La^D (\cdot | +)
\ee
are well-defined, translation-invariant and extremal elements of $\mathcal{G}(\gamma^D)$. For any $f$ bounded increasing, any other measure $\mu \in \mathcal{G}(\gamma^D)$ satisfies
\be \label{stochdom}
\mu^-[f] \leq \mu[f] \leq \mu^+[f].
\ee
For  longer ranges $1< \alpha \leq 2$, a phase transition holds for (\ref{LRDysonSpe}): There exists $\beta_c^D >0$ such that, for all $\beta > \beta_c^D$, we have
$\mu^- \neq \mu^+$   and moreover, at sufficiently low temperatures $\mathcal{G}(\gamma^D)=[\mu^-,\mu^+]$.
\end{proposition}


To get a candidate to represent the $g$-functions, i.e. the conditional probabilities w.r.t. the past, one needs to extend (\ref{LRDysonSpe}) to possibly {\em infinite} sets {$S$}, because the complement of the past  -- our future --  is infinite. Although we are far from the uniqueness regime, this has nevertheless been shown to be possible in our context following a general construction of \cite{FP}, made for  {\em attractive and right-continuous}\footnote{Right- or left-continuity corresponds to ``continuity in the direction $+$ or $-$'', see e.g. \cite{MRVM}.} specifications.
\begin{definition}\label{Glob}
A ``Global Specification'' $\Gamma$ on $\Z$ is a family of probability kernels $\Gamma=(\Gamma_S)_{S \subset \Z}$ on $(\Omega,\mathcal{F})$ from $\mathcal{F}_{S^c}$ to $\Omega_S$ such that for {\em any} $S$ subset of $\Z$:
\begin{enumerate}
\item $\Gamma_S(B|\omega)=\mathbf{1}_B(\omega)$ for all $\omega \in \Omega$ when $B \in \mathcal{F}_{S^c}$.
\item For all $S_1 \subset S_2 \subset \Z$, $\Gamma_{S_2} \Gamma_{S_1}=\Gamma_{S_2}$.
\end{enumerate}
 We write $\mu \in \mathcal{G}(\Gamma)$ if for all $A \in \mathcal{F}$ and {\em any} $S \subset \Z$,
\be \label{DLR4}
\mu[A|\mathcal{F}_{S^c}](\omega)=\Gamma_S(A|\omega), \; \mu{\rm -a.e.} \;  \omega.
\ee
\end{definition}

\begin{theorem}\cite{FP,ELN}\label{thmGlob}
Consider the Dyson model on $\Z$ at inverse temperature $\beta >0$, i.e. the specification $\gamma^D$ given by (\ref{LRDysonSpe}) and its extremal Gibbs  measure $\mu^+$ defined by (\ref{muplusminus}). A global specification  $\Gamma^+$  such that $\mu^+ \in \mathcal{G}(\Gamma^+)$  can be given as follows :
\begin{itemize}
\item For $S=\Lambda$ finite, for all $\omega \in \Omega$, set $\Gamma^+_\Lambda(d \sigma | \omega) := \gamma^D_\Lambda (d \sigma | \omega).$
\item For $S$ infinite, for all $\omega \in \Omega$, set $\Gamma^+_S(d\sigma | \omega):=\mu_S^{+,\omega} \otimes \delta_{\omega_{S^c}}(d \omega)$
where $\mu_S^{+,\omega}$ is the constrained measure on $(\Omega_S,\mathcal{F}_S)$ defined  as the (well-defined) weak limit  
$$
\mu_S^{+,\omega}(d \sigma_S):=\lim_{\Delta \uparrow S} \gamma^D_\Delta (d \sigma\mid +_S \omega_{S^c}).
$$
\end{itemize}
A similar construction yields a global specification $\Gamma^-$ so that $\mu^- \in \mathcal{G}(\Gamma^-)$.

\end{theorem}
These constructions allow us to consider, for given pasts, the expression of the  $g$-functions  as the magnetizations of Dyson models under various conditionings, and studying continuity will reduce to studying possible phase transition under constraints combined with  the study of  the stability of  interfaces.

Starting from $\mu^+$, we introduce $g^+$ to be the candidate to be the $g$-function representing (a version of) the single-site conditional probabilities  (\ref{g}) as a function of the past. Just as in  \cite{FP,ELN}, we introduce thus for any ``past'' configuration $\omega \in \Omega$:
$$
g^+(\omega):=\mu^+\left[\omega_0 | \mathcal{F}_{\{<0\}}](\omega)\right.
$$
Using the expression of Theorem \ref{thmGlob} in terms of global specifications and constrained measures with $S= \mathbb{Z}_+ = \{0,1,2,3,\ldots\}$, one gets, $\mu^+$-a.s. ($\omega$):

\begin{equation}\label{gplus}
g^+(\omega) =\Gamma_S^+[ \omega_0 | \omega] = \mu_S^{+,\omega} \otimes \delta_{\omega_{S^c}}[\omega_0]
\end{equation}
where $\mu_S^{+,\omega}$ is the constrained measure on $(\Omega_S,\mathcal{F}_S)$ as the (well-defined) weak limit
\be \label{constrLimit0}
\mu_S^{+,\omega}(d \sigma_S):=\lim_{\Delta \uparrow S} \gamma^D_\Delta (d \sigma\mid +_S \omega_{S^c}).
\ee

Previous works  and specific properties\footnote{Attractivity and right-continuity, see previous footnote and  also \cite{FP, ELN}.} insure $\mu^+$ is then indeed ``specified'' by $g^+$, in the sense that it is invariant by its left  action: $\mu^+ g^+=\mu^+$.

{\bf Note:} 
A non-continuous (= non-regular) $g$-function gives rise to a measure which is NOT a $g$-measure. To be a `proper'' $g$-function of the past, we would need that in addition to consistency, the function $g^+$ is {\em regular}, i.e.  essentially continuous (for which all  possible discontinuity points can be removed by modifications on  negligible sets).

Similarly to  e.g. \cite{ELN}, where two of us exhibited  (two-sided) discontinuity points by considering an alternating configuration $\omega_{{\rm alt}}$, we will  prove that for $L$ large and $N$ large compared to $L$, the  putative $g$-function $g^+$  can take significantly different values on  sub-neighborhoods $\mathcal{N}_{N,L}^{\pm,  left} \subset \mathcal{N}_{ L}(\omega_{{\rm alt}})$. Thanks to monotonicity-preservation, the constrained measure is explicitly built as the weak limit  (\ref{constrLimit}) obtained by $+$-boundary conditions fixed after a freezing 
$\omega$  in the past: For all $\omega \in \mathcal{N}_{L}(\omega_{{\rm alt}})$,
\be \label{constrLimit}
\mu^{+,\omega}_{\mathbb{Z}_+} (\cdot) =\lim_{I \in \mathcal{S},I \uparrow \mathbb{Z}_+} \gamma^D_I \left(\cdot\mid +_{(\mathbb{Z}_+)} \omega_{(\mathbb{Z}_+)^c}\right).
\ee

It is enough to consider this limit  along intervals  $I_n=[-n,+n] \cap \mathbb{Z}$ in the original space.

To disprove the $g$-measure property for the plus phase $\mu^+$ of our Dyson model, we will need to prove that a particular, in our case alternating, configuration $\omega_{{\rm alt}}$ is a non-removable point of discontinuity.  To do so, one has to find within its neighborhood two sub-neighborhoods (or at least two subsets of configurations of positive $\mu^+$ measure), on which  the value of $g^+$ drastically changes when modified arbitrarily far away. We  consider first finite-volume approximations of the constrained measure $\mu^{+,\omega}_{\mathbb{Z}^+}$ built as the weak limit (\ref{constrLimit}) with $+$-boundary condition by taking intervals $I_n$ arbitrarily large, larger than any other finite volumes encountered in this paper.

Consider the sub-neighborhoods $\mathcal{N}_{N,L}^{{ \pm},left}(\omega_{{\rm alt}})$ for $L<N<n$, whose size will be adjusted later. All together, this leads us to consider a partially frozen  Dyson model, either frozen into $+$ outside $I_n$, either into some arbitrary $\omega$ in $[-n,-N]$, or  into $-$  in the ''annulus'' $[-N-L, -L]$ and the alternating one $\omega_{{\rm alt}}$ in $[-L,-1]$.

\begin{center}
 \begin{tikzpicture}[]
    
\draw (-6,0) -- (6,0);

\node [below] at (-1.5,-0.2) {$-L$};

\node [below] at (-4.2,-0.2) {$-N-L$};

\node [below] at (5.5,-0.2) {$n$};

\node [below] at (-5.5,-0.2) {$-n$};

\foreach \n in {-3,...,3}{%
        \draw[fill] (\n,0) circle (1pt); 
    }
    
\foreach \n in {-6,...,6}{%
        \draw[fill] (\n/2,0) circle (1pt);    
    }

        
\foreach \n in {6,...,8}{%
        \draw[fill] (-\n/2,0) circle (1pt) node [above] {$-$};    
    }
    
    \foreach \n in {9,...,10}{%
        \draw[fill] (-\n/2,0) circle (1pt) node [above] {$\omega$};    
    }

    \foreach \n in {6}{%
        \draw[fill] (-\n,0) circle (1pt) node [above] {$+$}; 
    }
\foreach \n in {11}{%
        \draw[fill] (-\n/2,0) circle (1pt) node [above] {$+$};    
    }
    
       \foreach \n in {6}{%
        \draw[fill] (\n,0) circle (1pt) node [above] {$+$}; 
    }
\foreach \n in {11}{%
        \draw[fill] (\n/2,0) circle (1pt) node [above] {$+$};    
    }
    
       \foreach \n in {6}{%
        \draw[fill] (\n,-2) circle (1pt) node [above] {$+$}; 
    }
\foreach \n in {11}{%
        \draw[fill] (\n/2,-2) circle (1pt) node [above] {$+$};    
    }

\foreach \n in {1,2,...,3}{%
        \draw[fill] (\n+3,0) circle (1pt) node [above] {}; 
    }
    
\foreach \n in {1,2,...,6}{%
        \draw[fill] (\n/2+3,0) circle (1pt) node [above] {};    
    }
--------------

\node [below] at (0.5,-0.2) {$0$};

\foreach \n in {0,1,2,2.5}{%
        \draw[fill] (-\n,0) circle (1pt) node [above] {$-$}; 
    }
    
\foreach \n in {1,3}{%
        \draw[fill] (-\n/2,0) circle (1pt) node [above] {$+$};    
    }
    

\draw (-6,-2) -- (6,-2);
\node [below] at (0.5,-2.2) {$0$};
\node [below] at (-1.5,-2.2) {$-L$};

\node [below] at (-4.2,-2.2) {$-N-L$};

\node [below] at (5.5,-2.2) {$n$};
\node [below] at (-5.5,-2.2) {$-n$};

\foreach \n in {-3,...,3}{%
        \draw[fill] (\n,-2) circle (1pt); 
    }
    
\foreach \n in {-6,...,6}{%
        \draw[fill] (\n/2,-2) circle (1pt);    
    }

    
\foreach \n in {4,...,8}{%
        \draw[fill] (-\n/2,-2) circle (1pt) node [above] {$+$};    
    }
    \foreach \n in {9,...,10}{%
        \draw[fill] (-\n/2,-2) circle (1pt) node [above] {$\omega$};    
    }

\foreach \n in {6}{%
        \draw[fill] (-\n,-2) circle (1pt) node [above] {$+$}; 
    }
    
    \foreach \n in {6}{%
        \draw[fill] (-\n,-2) circle (1pt) node [above] {$+$}; 
    }
    
\foreach \n in {11}{%
        \draw[fill] (-\n/2,-2) circle (1pt) node [above] {$+$};    
    }

\foreach \n in {1,2,...,3}{%
        \draw[fill] (\n+3,-2) circle (1pt) node [above] {}; 
    }
    
\foreach \n in {1,2,...,6}{%
        \draw[fill] (\n/2+3,-2) circle (1pt) node [above] {};    
    }




 

\foreach \n in {0,1,2}{%
        \draw[fill] (-\n,-2) circle (1pt) node [above] {$-$}; 
    }
    
\foreach \n in {1,3,5}{%
        \draw[fill] (-\n/2,-2) circle (1pt) node [above] {$+$};    
    }
 
 \node[align=center, below] at (0,-3)%
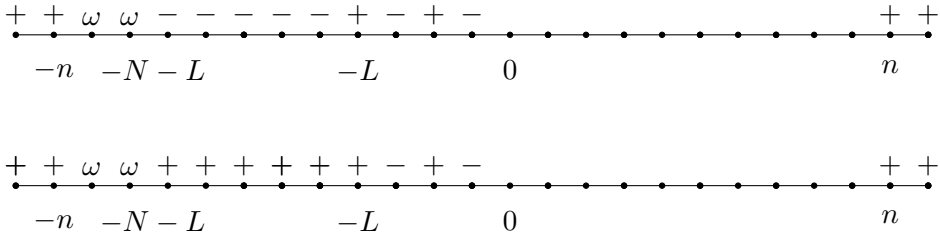
{Figure 1 : Left $\pm$ Neighborhoods of $\omega_{\rm alt}$};
 \end{tikzpicture}
\end{center}

By (\ref{gplus}),(\ref{constrLimit0}) and (\ref{constrLimit}), for a $\mu^+$-a.s. given $\omega$, the value taken by  $g^+$ will be the infinite-volume limit of  the magnetization of the finite-volume Gibbs measure of a Dyson-model on $[0,n]$, with the same decay $\alpha <2$ and $\omega$-dependent inhomogeneous external fields $h_x[\omega], x \geq 0$. In this minus case, for configurations $\omega : =\omega^-$ on the {sub-}neighborhood $\mathcal{N}_{N,L}^{-,left}(\omega_{{\rm alt}})$, one gets external fields 
$$
\forall x \geq 0,\; h_x[\omega]=\sum_{k=1}^L \frac{(-1)^k}{(k+x)^\alpha} - \sum_{k=L+1}^N \frac{1}{(k+x)^\alpha} + \sum_{k=N}^n \frac{\omega_{-k}}{(k+x)^\alpha} + 2 \sum_{k \geq n+1} \frac{1}{(k+x)^\alpha}
$$
while for $\omega:=\omega^+ \in \mathcal{N}_{N,L}^{+,left}(\omega_{{\rm alt}})$, we get :
$$
\forall x \geq 0,\; h_x[\omega]=\sum_{k=1}^L \frac{(-1)^k}{(k+x)^\alpha} + \sum_{k=L+1}^N \frac{1}{(k+x)^\alpha} + \sum_{k=N}^n \frac{\omega_{-k}}{(k+x)^\alpha} + 2 \sum_{k \geq n+1} \frac{1}{(k+x)^\alpha}
$$

We are reduced to study the magnetization under a  generalisation of the long-range RFIM (Random Field Ising Model), now  with a possibly dependent and/or biased, disordered external field, whose distribution is linked to the original measure $\mu$ itself via the distribution of the past. In such situations, when the fields are homogeneous one can sometimes use correlation inequalities and uniqueness via Lee-Yang \cite{LY} type arguments -- as were e.g. used to prove essential discontinuities for the decimation of Dyson model in \cite{ELN} -- but here our main difficulty is that this external field will change signs, depending of the value of $x \;  \in [0,1]$. For $n, L, N(L)$ large enough, it starts to be negative at  $0$ (due to its left-neighborhood frozen into minus in our alternating configuration) and, due to the $+$-boundary procedure far away, it becomes positive for $x$ large.

\begin{center}
 \begin{tikzpicture}[]
    
\draw (-6,0) -- (6,0);

\node [below] at (-2.5,-0.2) {$-L$};

\node [below] at (3.5,-0.2) {$n$};

\node [below] at (-5,-0.2) {$-n$};

\foreach \n in {-3,...,3}{%
        \draw[fill] (\n,0) circle (1pt); 
    }
    
\foreach \n in {-6,...,6}{%
        \draw[fill] (\n/2,0) circle (1pt);    
    }

    
\foreach \n in {6,...,10}{%
        \draw[fill] (-\n/2,0) circle (1pt) node [above] {$-$};    
    }

\foreach \n in {5}{%
        \draw[fill] (-\n,0) circle (1pt) node [above] {$+$}; 
    }
    
\foreach \n in {11}{%
        \draw[fill] (-\n/2,0) circle (1pt) node [above] {$+$};    
    }

\foreach \n in {1,2,...,3}{%
        \draw[fill] (\n+3,0) circle (1pt) node [above] {$+$}; 
    }
    
\foreach \n in {1,2,...,6}{%
        \draw[fill] (\n/2+3,0) circle (1pt) node [above] {$+$};    
    }



\node [above] at (0.25,0) {$\downarrow$}; 
 \node[align=center, above] at (0.25,0.5) { $h_x(\omega)<0$};

\node [above] at (2.5,0) {$\downarrow$}; 
\node[align=center, above] at (2.5,0.5) {$h_x(\omega)>0$};
 
\node [below] at (-0.5,-0.2) {$0$};

\foreach \n in {1,2}{%
        \draw[fill] (-\n,0) circle (1pt) node [above] {$-$}; 
    }
    
\foreach \n in {3,5}{%
        \draw[fill] (-\n/2,0) circle (1pt) node [above] {$+$};    
    }
    

\draw (-6,-2) -- (6,-2);

\node [below] at (-2.5,-2.2) {$-L$};

\node [below] at (3.5,-2.2) {$n$};

\node [below] at (-5,-2.2) {$-n$};

\foreach \n in {-3,...,3}{%
        \draw[fill] (\n,-2) circle (1pt); 
    }
    
\foreach \n in {-6,...,6}{%
        \draw[fill] (\n/2,-2) circle (1pt);    
    }

    
\foreach \n in {6,...,10}{%
        \draw[fill] (-\n/2,-2) circle (1pt) node [above] {$+$};    
    }

\foreach \n in {5}{%
        \draw[fill] (-\n,-2) circle (1pt) node [above] {$+$}; 
    }
    
\foreach \n in {11}{%
        \draw[fill] (-\n/2,-2) circle (1pt) node [above] {$+$};    
    }

\foreach \n in {1,2,...,3}{%
        \draw[fill] (\n+3,-2) circle (1pt) node [above] {$+$}; 
    }
    
\foreach \n in {1,2,...,6}{%
        \draw[fill] (\n/2+3,-2) circle (1pt) node [above] {$+$};    
    }



\node [above] at (0.25,-2) {$\downarrow$}; 
 \node[align=center, above] at (0.25,-1.5) {$h_x(\omega)>0$};

\node [above] at (2.5,-2) {$\downarrow$}; 
\node[align=center, above] at (2.5,-1.5) {$h_x(\omega)>0$};
 
\node [below] at (-0.5,-2.2) {$0$};

\foreach \n in {1,2}{%
        \draw[fill] (-\n,-2) circle (1pt) node [above] {$-$}; 
    }
    
\foreach \n in {3,5}{%
        \draw[fill] (-\n/2,-2) circle (1pt) node [above] {$+$};    
    }
 
 \node[align=center, below] at (0,-3)%
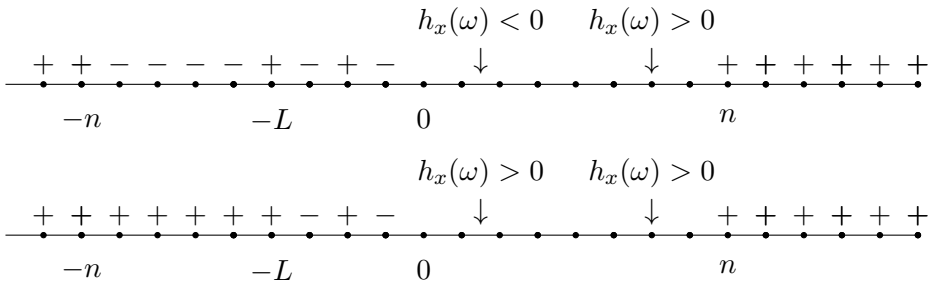
{Figure 2: Inhomogeneous $\omega$-dependent external fields};
 \end{tikzpicture}
\end{center}

Nevertheless, on the neighborhood $\mathcal{N}_{N,L}^{-,left}$, the inhomogeneous magnetic field $h_x(\omega)$ will stay negative  far enough to the past so that a $-$-phase is still felt at the origin in the  limits, while on the   neighborhood $\mathcal{N}_{N,L}^{+,left}$, a $+$-phase is always selected for $N$ and $L$ of adjusted size.
In the former case, we need to evaluate the effect of large, possibly huge, interval of minuses on its outside, faraway through an intermediate neutral interval, reminiscent of the phenomenon of {\em entropic repulsion in wetting phenomena} (see e.g. \cite{PV}, or \cite{Gia} for similar terminology in the setting of random polymers). To prove the essential discontinuity and in some sense "some" wetting beyond the origin through the alternating region, we first use the interface result of \cite{CMPR} (see also \cite{CMP17}) to state and prove in Section 3 a 
wetting result that we relate to entropic repulsion.



\subsection{Interfaces in Dyson models}

We will  thus derive our entropic repulsion argument from the interface result of  \cite{CMPR}. We start by describing and summarizing the latter and in particular briefly recall the contour construction based on triangles, that was first described in \cite{CFMP} to formalize the contour argument of \cite{frsP}.  Then we describe the Peierls estimate they obtain in this one-dimensional long-range context. In addition, this triangle construction also allows an unambiguous notion of interface in the phase transition region, as we describe now.

Let $L\ge 1$, and consider  $\Lambda=\Lambda_L=[-L,L]$. Define the dual lattice $\Lambda^*=\Lambda+\frac{1}{2}$ as the set  $\Lambda$ shifted by $1/2$. Given a configuration $\omega\in \{-1,+1\}^{\Lambda}$, let us define  configurations of triangles. A {\em spin-flip point} is a site $i$ in $\Lambda^*$ such that $\omega_{i-\frac{1}{2}}\neq \omega_{i+\frac{1}{2}}$. For each spin-flip point $i$, let us consider the interval $\left[i-\frac{1}{100},i+\frac{1}{100}\right]\subset \mathbb{R}$ and choose a real number $r_i$ in this interval such that, for every four distinct  points $r_{i_1}, r_{i_2}, r_{i_3}, r_{i_4}$ with $|r_{i_1}-r_{i_2}|\neq |r_{i_3}-r_{i_4}|$. The $r_i$ 's will be the bases of the triangles, and the last condition  is asked to avoid ambiguity in the construction of the triangle.

For each spin-flip point $i$, we start growing a ``$\lor$-line'' at $r_i$ where this $\lor$-line is embedded in $\mathbb{R}^2$ with angles $\pi/4$ and $3\pi/4$. If at some time two $\lor$-lines starting from different spin-flip points touch, the other two lines starting from those two spin-flip points stop growing,  and are removed without forming a triangle. Then we repeat this procedure. This process  can also be seen in the following way: for each $r_i$, draw a straight vertical line passing through it. Take the smallest distance between these lines, call the correponding $r_i$ and $r_j$ the spin-flip points of these lines, and draw a isosceles triangle with base angle $\pi/4$. Then, remove the lines associated to $r_i$ and $r_j$. Re-start.

Note that, for homogeneous boundary conditions, since the number of spin-flip points is even, every $r_i$ is a vertex of some triangle. On another hand, if we consider the Dobrushin boundary condition, then the number of spin-flip points is odd, and so there exists a {\em unique} spin-flip point which is not the vertex of any triangle. This point is called the ``interface point''. 

The first notion of interface point in this long-range one-dimensional context appeared in \cite{Joh2} in the terms of a ``thick interface'', and afterwards \cite{CMPR} defined the interface point according to the construction above.

Let
$$
T_L=\left\{ -1-\frac{1}{2L},-1+\frac{1}{2L},\ldots,-\frac{1}{2L},\frac{1}{2L},\ldots,1+\frac{1}{2L} \right\},
$$
and consider the Dobrushin boundary condition with all spins to the left of $\Lambda$ fixed to be minus and all spins to the right of $\Lambda$ fixed to be plus. Given a configuration $\omega$ in $\Lambda$, let  $I^*\equiv I^*(\omega)\in \Lambda^*$ be the interface point of the configuration $\omega$, and given $\theta\in T_L$, denote by
$$
\mathcal{S}_{\Lambda,\theta}=\{\omega: I^*=\theta L\}
$$
the set of spin configurations in $\Lambda$ for which the interface point is situated in $\theta L$. Note that this forms a  partition of $\Omega$ (if $\theta\neq \theta'$, then $\mathcal{S}_{\Lambda,\theta}\cap \mathcal{S}_{\Lambda,\theta'}=\emptyset$). We use it to define for each $\theta \in T_L$ the probability  to have an interface in $\theta L$ by
$$
\mu^{-+}_{\Lambda}[I^*=\theta L]=\frac{Z^{-+}_{\theta,\Lambda}}{Z^{-+}_{\Lambda}},
$$
where the partitions functions $
Z^{-+}_{\theta,\Lambda}=\sum_{\omega \in \mathcal{S}_{\Lambda,\theta}}e^{-\beta H_\Lambda^{-+}(\omega)}\; {\rm and} \;Z^{-+}_{\Lambda}=\sum_{\theta\in T_L}Z^{-+}_{\theta,\Lambda}
$ are defined via the Hamiltonian $H_\Lambda^{-+}$ in volume $\Lambda$ with Dobrushin boundary conditions.

For $i\in\Lambda$, the conditional expectation of $\omega_i$, given $I^*=\theta L$, is
$$
\mu^{-+}_{\theta,\Lambda}[\omega_i]:=\mu^{-+}_{\Lambda}[\omega_i|I^*=\theta L]=\frac{1}{Z^{-+}_{\theta,\Lambda}}\sum_{\omega \in \mathcal{S}_{\Lambda,\theta}}\omega_i e^{-\beta H_\Lambda^{-+}(\omega)}.
$$
Moreover, the expectation of $\omega_i$ in terms of in terms of  $\mu^{-+}_{\theta,\Lambda}[\omega_i]$ is
\begin{equation}\label{chain}
\mu^{-+}_{\Lambda_L}[\omega_i]=\sum_{\theta\in T_L}\mu^{-+}_{\theta,\Lambda_L}[\omega_i]\mu^{-+}_{\Lambda_L}(I^*=\theta L).
\end{equation}

These  constructions of triangles and associated contours are used in \cite{CMPR} to  get cluster expansions of partition functions that yield first the following proposition, which will be an essential tool for us. Let $Z^{-}_{\Lambda}$ be the partition function on $\Lambda$ with minus boundary condition, and let $\zeta(\alpha)=\sum_{k=1}^{\infty}\frac{1}{k^{\alpha}}$ be the Riemann zeta function. 

\begin{proposition}[Cassandro, Merola, Picco, Rozikov -- 2014]\label{prop:CMPR} 
For all $\alpha\in (\alpha_+,2)$, there exists $\beta_0\equiv \beta_0(\alpha) { >0}$ such that for all $\beta>\beta_0$ and $\theta \in T_L$, the following occurs:
$$
\begin{aligned}
&\log Z^{-+}_{\theta,\Lambda} - \log Z^{-}_{\Lambda}\\
&= -c_L(\alpha)L^{2-\alpha}+e^{-2\beta(\zeta(\alpha)+J)}\frac{L^{2-\alpha}}{(2-\alpha)(\alpha-1)}f_{\alpha}(\theta)(1\pm e^{-c_1(\alpha)\beta})(1+o(L)),
\end{aligned}
$$
where $f_{\alpha}(\theta)=(1+\theta)^{2-\alpha}+(1-\theta)^{2-\alpha}$, $c_L$ and $c_1$ are two positive constants depending on $\alpha$, once we require that the nearest-neighbor interaction $J=J(1)\gg 1$.
\end{proposition}
The restriction of $\alpha>\alpha_+$ appears since in \cite{CFMP} the proof of the phase transition of the Dyson model by a contour argument needs it\footnote{Although for the existence of a transition the validity can be extended to the whole range of phase-transition decays by FKG arguments. This does not work for  inhomogeneous situations such as disordered systems \cite{COP} or interface fluctuations \cite{CMPR}.}, while the contours introduced are based on the triangles  defined above.

 From this Proposition \ref{prop:CMPR}, one deduces in Corollary \ref{coro:CMPR} below that the interface point is located in the middle of the interval of $\Lambda$, up to a Gaussian correction which grows sublinearly in $L$. This means that the correction describes  mesoscopic fluctuations. In particular, this implies that macroscopic fluctuations are extremely improbable.

\begin{corollary}\label{coro:CMPR}
For every $\alpha\in (\alpha_+,2)$, there exists $\beta_1(\alpha)>\beta_0(\alpha)$ satisfying the following: for every $\beta>\beta_1(\alpha)$, there exists $\varepsilon\equiv \varepsilon(\beta)>0$ and $L(\varepsilon)\ge 1$ such that, for every $L>L(\varepsilon)$,
$$
\mu^{-+}_{\Lambda_L}\left[\left| I^* \right|> \varepsilon  L \right]\le 3(1-\varepsilon) Le^{-C L^{2-\alpha}(1+o(L))},
$$
where $C\equiv C(\alpha,\beta,\varepsilon)$ is a positive constant. 
\end{corollary}

\begin{proof}
Let $\alpha\in (\alpha_+,2)$  and $\theta \in T_L$. Differentiating $f_{\alpha}$ two times, we obtain
$$
\begin{aligned}
f'_{\alpha}(\theta)&=(2-\alpha)(1+\theta)^{1-\alpha}-(2-\alpha)(1-\theta)^{1-\alpha},\\
f''_{\alpha}(\theta)&=(2-\alpha)(1-\alpha)(1+\theta)^{-\alpha}+(2-\alpha)(1-\alpha)(1-\theta)^{-\alpha}.
\end{aligned}
$$
Thus, $f_{\alpha}$ on $[-1,1]$ only attains its maximum in $\theta=0$, and so $f_{\alpha}$ on $T_L$ attains its maximum in $\theta=\frac{1}{2}$ and $\theta=-\frac{1}{2}$.  Now, by Proposition \ref{prop:CMPR}, for every $\beta>\beta_0(\alpha)$,
\begin{equation}\label{eq:coro}
\begin{aligned}
&\log Z^{-+}_{\theta, \Lambda_L} - \log Z^{-+}_{\frac{1}{2},\Lambda_L} \\
&\le \frac{e^{-2\beta (\zeta(\alpha)+J)}}{(2-\alpha)(\alpha-1)} L^{2-\alpha}(1+o(L))\left[f_{\alpha}(\theta)(1+e^{-c_1(\alpha)\beta})-f_{\alpha}\left(\frac{1}{2}\right)(1-e^{-c_1(\alpha)\beta})\right].
\end{aligned}
\end{equation}
Since $f_{\alpha}(1)=f_{\alpha}(-1)<f_{\alpha}\left(\frac{1}{2}\right)$, there exists $\beta_1(\alpha)>\beta_0(\alpha)$ such that, 
$$
\forall \beta>\beta_1(\alpha),\; f_{\alpha}(1)=f_{\alpha}(-1)<f_{\alpha}\left(\frac{1}{2}\right)\left(\frac{1-e^{-c_1(\alpha)\beta}}{1+e^{-c_1(\alpha)\beta}}\right).
$$
Since $f_{\alpha}$ is continuous in $\theta\in [-1,1]$, there exists $\varepsilon=\varepsilon(\beta)>0$ and $L(\varepsilon)\ge 1$ such that, for every $L>L(\varepsilon)$ with $\{\theta \in T_L:|\theta|\le \varepsilon\}\neq \emptyset$, we have that, for $\theta\in T_L$ with $|\theta|>\varepsilon$,
\begin{equation}\label{eq:coro_f}
f_{\alpha}(\theta)<f_{\alpha}\left(\frac{1}{2}\right)\left(\frac{1-e^{-c_1(\alpha)\beta}}{1+e^{-c_1(\alpha)\beta}}\right).
\end{equation}
Let us define $M_{\alpha}=\max\{f_{\alpha}(\theta):|\theta|>\varepsilon\}$ and
\begin{equation}\label{eq:coro_g}
g(\alpha,\beta)=\frac{e^{-2\beta(\zeta(\alpha)+J)}}{(2-\alpha)(\alpha-1)} \left[f_{\alpha}\left(\frac{1}{2}\right)(1-e^{-c_1(\alpha)\beta})-M_{\alpha}(1+e^{-c_1(\alpha)\beta})\right].
\end{equation}
Thus, for every $\beta>\beta_1(\alpha)$ and $L>L(\varepsilon)$, from (\ref{eq:coro}), (\ref{eq:coro_f}) and (\ref{eq:coro_g}), we have
$$
\begin{aligned}
\mu^{-+}_{\Lambda_L}\left[\left| I^* \right|> \varepsilon L  
  \right]&= \frac{1}{Z^{-+}_{\Lambda}}\sum_{\substack{\theta\in T_L\\\left| \theta \right|> \varepsilon  }}   
   Z^{-+}_{\theta,\Lambda} \\
     & \le \frac{Z^{-+}_{\frac{1}{2},\Lambda}}{Z^{-+}_{\Lambda}}\sum_{\substack{\theta\in T_L\\\left| \theta \right|> \varepsilon  }}   
   e^{-g(\alpha,\beta)L^{2-\alpha}(1+o(L))}. \\
   &\le |\{\theta \in T_L: |\theta|>\varepsilon\}|e^{-g(\alpha,\beta)L^{2-\alpha}(1+o(L))}\\
   &\le 3(1-\varepsilon)L e^{-g(\alpha,\beta)L^{2-\alpha}(1+o(L))},
\end{aligned}
$$
as we desired.
\end{proof}

From Proposition \ref{prop:CMPR} and the observation that, at finite volume $\Lambda$, for any $x^* \in \Lambda$,
$$
\mu^{-+}_{\theta,\Lambda}[\omega_{x^*}]= \frac{d}{dg} (\log{Z_{\theta, \Lambda}^{g,x^*}}) \Big|_{g=0}
$$
where  for any $g \in \R$, $Z_{\theta, \Lambda}^{g,x^*}=\sum_{\sigma_\Lambda \in \mathcal{S}_{\Lambda, \theta}} e^{-\beta H_{\Lambda}^{-+}(\sigma_\Lambda) + g \sigma_{x^*}}$, Cassandro {\em et al.} also obtained in \cite{CMPR} the following estimate for  important conditional magnetizations, which will provide our first step towards wetting and entropic repulsion in the next section:

\begin{proposition}[Cassandro, Merola, Picco, Rozikov -- 2014]\label{prop2:CMPR}
For all $\alpha\in (\alpha_+,2]$, there exists $\beta_0\equiv \beta_0(\alpha)$ such that for all $\beta>\beta_0$,  $
\mu^{-+}_{\theta,\Lambda}[\omega_i] = \pm 1 \ \text{ if }i=\theta L \pm \frac{1}{2}$ and
$$
\begin{aligned}
\mu^{-+}_{\theta,\Lambda}[\omega_i]=
&\left[ 1-2e^{-2\beta( \zeta(\alpha)+J)}e^{\frac{2\beta}{\alpha-1}\frac{1}{|i-\theta L|^{\alpha-1}}}\left[ 1+\mathcal{O}(e^{-c_1\beta}) \right]\left[ 1+o\left( \frac{1}{L} \right) \right] \right]\\
&\times \left[ \mathbbm{1}_{i>\theta L+\frac{1}{2}} - \mathbbm{1}_{i<\theta L - \frac{1}{2}} \right].
\end{aligned} 
$$
\end{proposition}

\section{Entropic repulsion --  Wetting transition}

For a fixed $N> 1$, we will consider the plus phase $\mu^{+}$, conditioned on the event $-_{-N,-1}$ of there being an interval $[-N,-1]$ of minus spins. We claim that there are two intervals of length of order  $L$, namely $[-N-\frac{(1-\varepsilon)}{2}L,-N-1]$, and $[0,\frac{(1-\varepsilon)}{2}L]$ (where $\varepsilon$ is from Corollary \ref{coro:CMPR}) left and right of the fixed interval, such that for $N \gg L$ both large enough, the magnetization of the spins in $\Delta$ conditioned on the event $-_{-N,-1}$ is negative, whenever $\Delta$ is in one of those intervals. These intervals play the role of a ``completely wet region'' in a wetting transition\footnote{Note that this wetting is a positive-temperature effect. Indeed, at zero temperature the interface with Dobrushin boundary conditions is homogeneously distributed, and a frozen interval of minuses, inserted in a plus configuration, will have only pluses to the left and to the right.} (See Figure 3). In other words,

\begin{proposition}\label{prop}

Let $\alpha\in (\alpha_+,2)$ and $\beta_1\equiv \beta_1(\alpha)$ from Corollary \ref{coro:CMPR}. Then, there exists $\beta_2>\beta_1$  such that for any $\beta>\beta_2$, there exists $\varepsilon=\varepsilon(\beta)$, $L_0(\alpha,\beta)>L(\varepsilon)$ from Corollary \ref{coro:CMPR} such that, for any $L>L_0$, if $N > L$ and $LN^{1-\alpha}=o(1)$ then,
$$
\mu^{+} [\omega_i| -_{-N,-1}] \leq - \frac{m}{2} 
$$
for a suitable $m=m(\beta)>0$ and for every $i\in [-N-\frac{(1-\varepsilon)}{2}L,-N-1]\cup [0,\frac{(1-\varepsilon)}{2}L]$.
\end{proposition}

\begin{proof}
Fix $\alpha\in (\alpha_+,2)$ and $\beta_0\equiv \beta_0(\alpha)$ from Proposition \ref{prop:CMPR}. We will first prove the statement for $i\in [0,\frac{(1-\varepsilon)}{2}L]$.
  
The main idea of our proof is to choose $N$ large enough for the total influence of all spins left of the interval to be bounded by a (small) constant, so that one can neglect boundary effects beyond $-N$ as in \cite{BLP}. Then inside the interval of length $L$, the interface separating the plus and minus phases is with large probability within the same window as with the Dobrushin boundary conditions. If afterwards we move the plus-boundary to the right, the location of the interface can also move only to the right, that is away from the frozen interface (by an FKG argument).

To make this precise we proceed as follows. 
From Corollary \ref{coro:CMPR}, if we would consider the interval $\tilde{\L}_L:=[0,2L]$ with Dobrushin boundary condition, the interface point will with overwhelming probability lie about halfway, with fluctuations which are ``mesoscopic'', that is, there exist{s} $\beta_1>\beta_0$, such that, for $\beta>\beta_1$ there exist $\varepsilon=\varepsilon(\beta)<1$ and $L(\varepsilon)\ge 1$ in which, for every $L>L(\varepsilon)$,
\begin{equation}\label{desloc}
\mu_{\tilde{\L}_L}^{-+}\left[\left| I^* - L \right|> \varepsilon  L \right]\le 3(1- \varepsilon) Le^{-C L^{2-\alpha}(1+o(L))},
\end{equation}
Let us take $i\in \Delta_{\varepsilon,L}:=[0,\frac{(1-\varepsilon)}{2}L]$. Note that, for every $\theta \in [1-\varepsilon,1+\varepsilon]$, we have $|i-\theta L|^{1-\alpha}\le [\frac{(1-\varepsilon)}{2}L]^{1-\alpha}$. Thus, for every $0<\delta<1$, there exists $\beta_2\equiv \beta_2(\alpha,\delta)>\beta_1$ such that, for every $\beta>\beta_2$, there exists $L(\alpha,\beta,\varepsilon,\delta)>L(\varepsilon)$ so that, by Proposition \ref{prop2:CMPR}, for every $L>L(\alpha,\beta,\varepsilon,\delta)$, we have $\mu^{-+}_{\theta,\tilde{\Lambda}_L}[\omega_i]<-1+\delta$ for every $i\in \Delta_{\varepsilon,L}$.
By (\ref{chain}) and (\ref{desloc}), for every $i\in \Delta_{\varepsilon,L}$,
\begin{equation}\label{eq:magnetization}
\begin{aligned}
\mu^{-+}_{\tilde{\Lambda}_L}[\omega_i]&=\sum_{1-\varepsilon\le \theta \le 1+\varepsilon}\mu^{-+}_{\theta,\tilde{\Lambda}_L}[\omega_i]\mu^{-+}_{\tilde{\Lambda}_L}(I^*=\theta L)+\sum_{\substack{\theta<1-\varepsilon \\ \theta>1+\varepsilon}}\mu^{-+}_{\theta,\tilde{\Lambda}_L}[\omega_i]\mu^{-+}_{\tilde{\Lambda}_L}(I^*=\theta L)\\
&\le \sum_{1-\varepsilon\le \theta \le 1+\varepsilon}\mu^{-+}_{\theta,\tilde{\Lambda}_L}[\omega_i]\mu^{-+}_{\tilde{\Lambda}_L}(I^*=\theta L)+ 3(1-\varepsilon)Le^{-CL^{2-\alpha}(1+o(L))}\\
&<(-1+\delta)(1-3(1-\varepsilon)Le^{-CL^{2-\alpha}(1+o(L))})+3(1-\varepsilon)Le^{-CL^{2-\alpha}(1+o(L))}\\
&<-1+\eta,
\end{aligned}
\end{equation}
for some $0<\eta<1$ and for every $L$ sufficiently large.
For any $N>1$, if we lift the constraint that all spins are minus to the left of site $-N$, the total energy due to the boundary condition changing inside the interval $\Delta_{\varepsilon,L}$ is bounded by
\begin{equation}\label{prop:energy}
\sum_{j< -N}\sum_{i=0}^{2L} \frac{1}{|i-j|^{\alpha}}\omega_i\omega_j \le (2L+1)\sum_{i> N} \frac{1}{i^{\alpha}}\le 3L\int_{N}^{+\infty}\frac{1}{x^{\alpha}}\d x \le \frac{3}{\alpha-1}LN^{1-\alpha}
\end{equation}
which remains bounded since we are assuming $LN^{1-\alpha}=o(1)$ and, in particular, for $i\in \Delta_{\varepsilon,L}$, by (\ref{eq:magnetization}), we have
\begin{equation}
\mu_{{\L}^{'}_L}^{+}[\omega_i| -_{-N,-1}] \leq 
e^{3\beta cLN^{1-\alpha}} \mu_{\tilde{\L}_L}^{-+}(\omega_i=1) - \mu_{\tilde{\L}_L}^{-+}(\omega_i=-1) < -\frac{m}{2}
\end{equation}
where ${\L}^{'}_L = [-N, 2L]$ and $m:=\mu^{+}[\omega_i]$.

Due to the FKG property, for any $\L$ such that ${\L}^{'}_L \subset \L$, we have
$$
\mu_{\L}^{+}[\omega_i| -_{-N,-1}] \leq  \mu_{{\L}^{'}_L}^{+}[\omega_i| -_{-N,-1}],
$$
for all $i\in \Delta_{\varepsilon,L}$. Therefore, for any site $i \in [0,\frac{(1-\varepsilon)}{2} L]$, there exists $L_0\ge 1$ such that, for $L>L_0$ and $LN^{1-\alpha}=o(1)$, $\mu^{+}[ \omega_i| -_{-N,-1}] < - \frac{m}{2}.$

For the wetting of sites $i$ in the other interval $[-N-\frac{(1-\varepsilon)}{2}L,-N-1]$, we consider the Gibbs measure with reverse Dobrushin boundary condition $\mu^{+-}$, i.e., $\omega_i=1$ if $i<0$, and $\omega_i=-1$ if $i\ge 0$, and apply the same argument as above. Thus, for $N$ large enough,
$$
\mu^{+}[\omega_i| -_{1,N}] <-\frac{m}{2}
$$
for every $i \in [-\frac{(1-\varepsilon)}{2}L,0]$, where $-_{1,N}$ is the event of there being an interval $[1,N]$ of minus spins. Since the Dyson model is translational invariant, when we shift all sites by $-N$, we are done.
\end{proof}

%

\begin{center}
 \begin{tikzpicture}[]
    
\draw (-6,0) -- (6,0);

\node [below] at (-6,-0.2) {$-N-\frac{(1-\varepsilon)}{2}L$};
\node [below] at (-1.5,-0.2) {$-N$};
\node [below] at (1.5,-0.2) {$0$};
\node [below] at (5.5,-0.2) {$\frac{(1-\varepsilon)}{2}L$};

\foreach \n in {-5,...,5}{%
       \draw[fill] (\n,0) circle (1pt);
    }
    
\foreach \n in {-11,...,11}{%
       \draw[fill] (\n/2,0) circle (1pt);
    }


\draw [decoration={brace,amplitude=0.4cm}, decorate] (-5.5,0) -- (-2,0);
\node[align=center, above] at (-4,0.5) {$-$ phase};

\foreach \n in {-1,0,1}{%
        \draw[fill] (\n,0) node [above] {$-$}; 
    }
    
\foreach \n in {-3,-1,1}{%
        \draw[fill] (\n/2,0) node [above] {$-$};    
    }

\draw [decoration={brace,amplitude=0.4cm}, decorate] (1.5,0) -- (5.5,0);
 \node[align=center, above] at (3.5,0.5) {$-$ phase};


 \node[align=center, below] at (0,-1)%
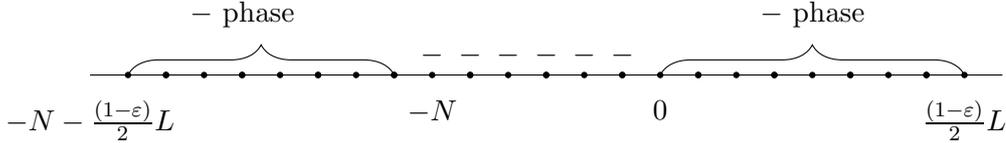
{Figure 3 : {wetting transition at low temperature}};
 \end{tikzpicture}
\end{center}

\section{Lack of the $g$-measure property: proof}

In this section, we provide the proof of Theorem \ref{thm:main}.

The main idea is first to decouple the spins in a subinterval $[1,L_1]$ of the ``wet'' minus interval of length $o(L)$, such that $L_1$ is large, but small compared to $L$. As the energy difference due to the decoupling is small compared to the energy cost of moving the interface, the location of the interface as analyzed in  \cite{CMPR} does not change, when viewed  on scale $L$.
If then, in the next step, the decoupled region is frozen in an alternating configuration and recoupled, this causes an extra finite-energy term --as compared to being decoupled--, which again will hardly influence the location of the interface (and thus the size of the wet region). 

%

Let us first present a lemma.
\begin{lemma}\label{lema-aernout}
Let $\alpha \in (1,2)$ and $L_1> 1$. Consider the observable in $\Omega$ given by
\begin{equation}\label{eqB}
B(\omega)=\sum_{j\notin [-L_1,-1]}\sum_{i\in [-L_1,-1]}\frac{(-1)^i}{|i-j|^{\alpha}}\omega_j.
\end{equation}
Then, there exists $c>0$ such that $\sup_{\omega} |B(\omega)| = \lVert B \rVert\le c$, where $c$ does not depend on $L_1$.
\end{lemma}

\begin{proof}
Let us find a configuration $\omega$ which attains the maximum of the sum (\ref{eqB}).
Note that
$$
\sum_{j\notin [-L_1,-1]}\sum_{i\in [-L_1,-1]}\frac{(-1)^i}{|i-j|^{\alpha}}\omega_j 
= \underbrace{\sum_{j\ge 0}\sum_{i=1}^{L_1}\frac{(-1)^i}{|i+j|^{\alpha}}\omega_j}_{(a)}+\underbrace{\sum_{j\ge 1}\sum_{i=1}^{L_1}\frac{(-1)^i}{|i-L_1-j|^{\alpha}}\omega_{-L_1-j}}_{(b)}.
$$

For $(a)$, note that, for each $j\ge 0$, the sum 
$\sum_{i=1}^{L_1}\frac{(-1)^i}{|i+j|^{\alpha}}\omega_j$
is positive whenever $\omega_j=-1$. Thus, $\omega_j=-1$ for every $j\ge 0$.

For $(b)$, for each $j\ge 0$, the sum
$\sum_{i=1}^{L_1}\frac{(-1)^i}{|i-L_1-j|^{\alpha}}\omega_{-L_1-j}$
is positive if satisfies the following,
\begin{itemize}
\item If $L_1$ is even, then $\omega_j=1$ for every $j<-L_1$;
\item If $L_1$ is odd, then $\omega_j=-1$ for every $j<-L_1$.
\end{itemize}
Thus, the configuration $\omega$ constructed above makes $(a)$ and $(b)$ be positive, and so maximizes $B(\omega)$.

Now, from $(a)$, let us prove that there exists $c_1>0$ such that
\begin{equation}\label{eq_c1}
\sum_{j\ge 0}\sum_{i=1}^{L_1}\frac{(-1)^{i+1}}{|i+j|^{\alpha}}\le c_1
\end{equation}
is summable and $c_1$ does not depend on $L_1$. Define
$$
R_N(\alpha)=\sum_{n>N}\frac{(-1)^{n+1}}{n^{\alpha}}.
$$
Then, for a fixed $j\ge 0$,
\begin{equation}\label{eq_R}
\sum_{i=1}^{L_1}\frac{(-1)^{i+1}}{|i+j|^{\alpha}} = R_{j}(\alpha) - R_{L_1+j}(\alpha).
\end{equation}
We have, for each $N\ge 0$,
$$
|R_{2N}(\alpha)| \le  \sum_{n\ge N+1}\int_{2n-1}^{2n} \alpha x^{-\alpha-1}\d x
< \frac{1}{(2N+1)^{\alpha}}.
$$
Also, with the same argument,
$$
|R_{2N+1}(\alpha)| \le \frac{1}{(2N+2)^{\alpha}}.
$$
Thus, for any $N\ge 1$, we have $|R_{N}(\alpha)| \le (N+1)^{-\alpha}$, then,
\begin{equation}\label{eq2_R}
R_{j}(\alpha) - R_{L_1+j}(\alpha) \le (j+1)^{-\alpha} + (L_1+j+1)^{-\alpha}.
\end{equation}
Since $\sum_{n\ge 1}n^{-\alpha}$ is summable, we have that the sum (\ref{eq_R}) is summable. Therefore, there exists $c_1>0$ such that (\ref{eq_c1}) holds.
Moreover, since the right hand side of (\ref{eq2_R}) decreases when $L_1$ increases, the constant $c_1$ does not depend on $L_1$.

For $(b)$, we have, for a fixed $j\ge 1$ and $L_1$ even,
$$
\begin{aligned}
\sum_{i=1}^{L_1}\frac{(-1)^{i}}{|i-L_1-j|^{\alpha}}
&=R_{L_1+j-1}(\alpha) - R_{j-1}(\alpha)
&\le \frac{1}{(L_1+j)^{\alpha}}+\frac{1}{j^{\alpha}}.
\end{aligned}
$$
For $L_1$ odd, the argument is the same,
$$
\sum_{i=1}^{L_1}\frac{(-1)^{i+1}}{|i-L_1-j|^{\alpha}} \le \frac{1}{(L_1+j)^{\alpha}}+\frac{1}{j^{\alpha}}.
$$
Thus, there exists $c_2>0$ such that
\begin{equation}\label{eq_c2}
\sum_{j\ge 1}\sum_{i=1}^{L_1}\frac{(-1)^i}{|i-L_1-j|^{\alpha}}\omega_{-L_1-j} \le c_2.
\end{equation}
By (\ref{eq_c1}) and (\ref{eq_c2}), we conclude the proof.

\end{proof}

\begin{proof}[Proof of Lemma \ref{thm:main}]
For a fixed $\alpha\in (\alpha_+,2)$ and $L_1>1$, let us consider the interaction set $\Upsilon_{L_1}=\{\{i,j\}\in \mathbb{Z}^2:i\neq j,\ \{i,j\}\subset [-L_1,-1] \text{ or }\{i,j\}\cap [-L_1,-1]=\emptyset\}$, i.e., we remove the interactions between $[-L_1,-1]$ and its complement. For a finite subset $\Lambda$ containing $[-L_1,-1]$ denote the Hamiltonian
\begin{equation}\label{hamilton}
H^{\tau}_{\Lambda,1}(\omega)=-\sum_{\substack{\{i,j\}\in \Upsilon_{L_1}\\ i,j\in \Lambda}}|i-j|^{-\alpha} \omega_i\omega_j - \sum_{\substack{\{i,j\}\in \Upsilon_{L_1}\\ i\in \Lambda, j\notin \Lambda}}|i-j|^{-\alpha} \omega_i\tau_j,
\end{equation}
where $\tau$ is a boundary condition. Denote by $\mu_{\Lambda,1}^{\tau}$ be  corresponding Gibbs measures

$$
\mu_{\Lambda,1}^{\tau}(\omega)=\frac{1}{Z^{\tau}_{\Lambda,1}}e^{-\beta H^{\tau}_{\Lambda,1}(\omega)}.
$$
Note that the cost of the total energy to remove these bonds is bounded by
\begin{equation}\label{thm:eq1}
\left|H^{\tau}_{\Lambda}(\omega)- H^{\tau}_{\Lambda,1}(\omega) \right|=\left|\sum_{\substack{j<-L_1 \\ j>-1}}\sum_{-L_1\le i\le -1} |i-j|^{-\alpha}\omega_i\omega_j\right| \le c L_1^{2-\alpha}
\end{equation}
for every finite subset $\Lambda$ containing $[-L_1,-1]$, for some constant $c>0$.  Consider $\beta > \beta_{2}$ from Proposition \ref{prop}, $L=L(L_1)$ satisfying $L_1=o(L)$, and the interval $\Delta_{2L}=[-L_1,2L-L_1]$. By Corollary \ref{coro:CMPR} and (\ref{thm:eq1}), we have
\begin{equation}\label{eq:ineq}
\begin{aligned}
\mu_{\Delta_{2L,1}}^{-+}\left[\left| I^* - \left(L - L_1\right) \right|> \varepsilon  L \right] \leq  3(1-\varepsilon) Le^{-C L^{2-\alpha}(1+o(L))+\beta o(L^{2-\alpha})}
\end{aligned}
\end{equation}
for some constant $C>0$. Hence the location of the interface point will not be majorly effected and, by Proposition \ref{prop2:CMPR}, we have $\mu_{\Delta_{2L,1}}^{-+}[\omega_i] <-m/2$ for every $i\in \Delta_L=[-L_1,\frac{(1-\varepsilon)}{2}L-L_1]$.

Using the same argument as in Proposition \ref{prop}, for $\beta_3\equiv \beta_3(\alpha)$ and $\beta>\beta_3$, for $L$ with $L_1=o(L)$ and $N> N(L)$ such that $LN^{1-\alpha}=o(1)$, the magnetization of each spin in $\left[0,\frac{(1-\varepsilon)}{2}L-L_1\right]\cup [-N-L_1-\frac{(1-\varepsilon)}{2}L,-N-L_1-1]$ is negative when we constrain the frozen interval $[-N-L_1,-L_1-1]$ to be minus, i.e., considering $\Lambda''_L=[-N-L_1,2L-L_1]$,
\begin{equation}
\begin{aligned}
\mu^+_{\Lambda''_{L,1}}[\omega_i| -_{-N-L_1,-L_1-1}] &\le e^{3\beta cLN^{1-\alpha}}\mu^{-+}_{\Delta_{2L,1}} [\omega_i=1] - \mu^{-+}_{\Delta_{2L,1}} [\omega_i=-1] \\
&< -\frac{m}{2}.\\
\end{aligned}
\end{equation}

Now, denote by $A_{L_1}$ the set of configurations that are alternating in $[-L_1,-1]$. Since 
$$
\mu^+_{\Lambda''_{L,1}}[\omega_i| -_{-N-L_1,-L_1-1}] = \mu^+_{\Lambda''_{L,1}}[\omega_i| -_{-N-L_1,-L_1-1}\cap A_{L_1}],
$$
for every $i\in \left[0,\frac{(1-\varepsilon)}{2}L-L_1\right]\cup [-N-L_1-\frac{(1-\varepsilon)}{2}L,-N-L_1-1]$, by FKG inequality, we have
\begin{equation}
\mu^+_1 [\omega_i| -_{-N-L_1,-L_1-1}\cap A_{L_1}] \leq  - \frac{m}{2}.
\end{equation}
Thus, the spins in the set $\left[0,\frac{(1-\varepsilon)}{2}L-L_1\right]\cup [-N-L_1-\frac{(1-\varepsilon)}{2}L,-N-L_1-1]$ are in the minus phase.

By the same argument, considering all spins in the frozen interval $[-N-L_1,-L_1-1]$ being plus, then the spins in the set $\left[0,\frac{(1-\varepsilon)}{2}L-L_1\right]\cup [-N-L_1-\frac{(1-\varepsilon)}{2}L,-N-L_1-1]$ are in the plus phase (see Figure 4). In particular, 
\begin{equation}
\mu^+_1 [\omega_0| -_{-N-L_1,-L_1-1}\cap A_{L_1}] \leq  - \frac{m}{2} <0 < \frac{m}{2} \leq \mu^+_1 [\omega_0| +_{-N-L_1,-L_1-1}\cap A_{L_1}]. 
\end{equation}

The measures $\mu^+_1 [ \cdot | -_{-N-L_1,-L_1-1}\cap A_{L_1}]$ and  $\mu^+_1 [ \cdot | +_{-N-L_1,-L_1-1}\cap A_{L_1}]$ are FKG measures (satisfying the FKG inequality). This fact is a consequence of the Holley inequality and, in addition, these measures are extremal Gibbs measures associated to the Hamiltonian (\ref{hamilton}).

By Lemma \ref{lema-aernout}, the sum of the interaction terms between $[-L_1,-1]$ and its complement is uniformly bounded by a constant. Then,
if we insert back the interactions connecting with $\Upsilon_{L_1}$, 
this changes the Hamiltonian by a uniformly bounded (finite-energy) term. 

Using a Bricmont--Lebowitz--Pfister type argument as in \cite{BLP}, we can show that conditional probabilities with respect to the original measures
$\mu^+ [ \cdot | -_{-N-L_1,-L_1-1}\cap A_{L_1}]$ and  $\mu^+ [ \cdot | +_{-N-L_1,-L_1-1}\cap A_{L_1}]$, associated to the of the Dyson model, are equivalent to $\mu^+_1 [ \cdot | -_{-N-L_1,-L_1-1}\cap A_{L_1}]$ and  $\mu^+_2 [ \cdot | +_{-N-L_1,-L_1-1}\cap A_{L_1}]$ respectively, and then they are also different extremal Gibbs measures. In addition, 
\begin{equation}
\mu^+ [\omega_0| -_{-N-L_1,-L_1-1}\cap A_{L_1}] < \mu^+ [\omega_0| +_{-N-L_1,-L_1-1}\cap A_{L_1}].
\end{equation}
Thus, for every configuration $\omega^+ \in \mathcal{N}_{N+L_1,L_1}^{+,left}(\omega_{\rm alt})$ and $\omega^- \in \mathcal{N}_{N+L_1,L_1}^{-,left}(\omega_{\rm alt})$, there exists $\delta>0$ such that
$$
\left| \mu^{+,\omega^+}_{\mathbb{Z}_+}[\sigma_0] -  \mu^{+,\omega^-}_{\mathbb{Z}_+}[\sigma_0] \right| > \delta,
$$
for $L_1$ large enough, as we desired.
\end{proof}

\begin{center}
 \begin{tikzpicture}[]
    
\draw (-6,0) -- (6,0);

\node [below] at (-6.3,-0.2) {$-N-L_1-\frac{(1-\varepsilon)}{2}L$};


\node [below] at (-3,-0.2) {$-N-L_1$};

\node [below] at (-1.7,-0.2) {$\cdots$};

\node [below] at (-1,-0.2) {$-L_1$};

\node [below] at (2,-0.2) {$0$};


\node [below] at (5.5,-0.2) {$-L_1+\frac{(1-\varepsilon)}{2}L$};

\foreach \n in {-5,...,-3,-1,0,1,2,3,4,5}{%
       \draw[fill] (\n,0) circle (1pt);
    }
    
\foreach \n in {-11,...,-5,-3,-2,-1,0,1,2,3,4,5,6,7,8,9,10,11}{%
       \draw[fill] (\n/2,0) circle (1pt);
    }

    
\draw [decoration={brace,amplitude=0.3cm}, decorate] (-5.5,0) -- (-3,0);
 \node[align=center, above] at (-4.3,0.4) {$-$ phase};


\draw [decoration={brace,amplitude=0.25cm}, decorate] (-2.5,0) -- (-1.5,0);
 \node[align=center, above] at (-2,0.4) {$-$};
 


\foreach \n in {-1,...,1}{%
        \draw[fill] (\n,0) node [above] {$+$}; 
    }
    
\foreach \n in {-1,1,3}{%
        \draw[fill] (\n/2,0) node [above] {$-$};    
    }
    

\draw [decoration={brace,amplitude=0.4cm}, decorate] (2,0) -- (5.5,0);
 \node[align=center, above] at (3.7,0.4) {$-$ phase};
    

\draw (-6,-2.5) -- (6,-2.5);

\node [below] at (-6.3,-2.7) {$-N-L_1-\frac{(1-\varepsilon)}{2}L$};


\node [below] at (-3,-2.7) {$-N-L_1$};

\node [below] at (-1.7,-2.7) {$\cdots$};

\node [below] at (-1,-2.7) {$-L_1$};

\node [below] at (2,-2.7) {$0$};


\node [below] at (5.5,-2.7) {$-L_1+\frac{(1-\varepsilon)}{2}L$};

\foreach \n in {-5,...,-3,-1,0,1,2,3,4,5}{%
       \draw[fill] (\n,-2.5) circle (1pt);
    }
    
\foreach \n in {-11,...,-5,-3,-2,-1,0,1,2,3,4,5,6,7,8,9,10,11}{%
       \draw[fill] (\n/2,-2.5) circle (1pt);
    }

    
\draw [decoration={brace,amplitude=0.3cm}, decorate] (-5.5,-2.5) -- (-3,-2.5);
 \node[align=center, above] at (-4.3,-2.1) {$+$ phase};


\draw [decoration={brace,amplitude=0.25cm}, decorate] (-2.5,-2.5) -- (-1.5,-2.5);
 \node[align=center, above] at (-2,-2.1) {$+$};
 


\foreach \n in {-1,...,1}{%
        \draw[fill] (\n,-2.5) node [above] {$+$}; 
    }
    
\foreach \n in {-1,1,3}{%
        \draw[fill] (\n/2,-2.5) node [above] {$-$};    
    }
    

\draw [decoration={brace,amplitude=0.4cm}, decorate] (2,-2.5) -- (5.5,-2.5);
 \node[align=center, above] at (3.7,-2.1) {$+$ phase};
 
 \node[align=center, below] at (0,-3.3)%
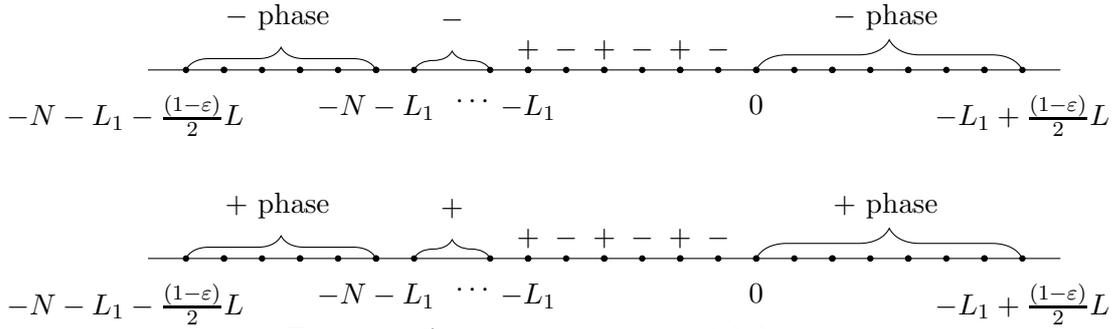
{Figure 4 :  from wetting to essential discontinuity. \\ Here $L_1=o(L)$ and $LN^{1-\alpha}=o(1)$};
 \end{tikzpicture}
\end{center}

\section{Final remarks and  open questions}

We have as our main result shown that between  the class of Gibbs measures and the class of $g$-measures, neither of them contains the other one. Thus one-sided continuity and two-sided continuity of conditional probabilities are really different properties and there exists a clear distinction between these two notions.

The result on entropic repulsion which we used in the proof presumably can be improved in various respects. We mention a few open questions regarding these issues.

It is not clear to us whether entropic repulsion holds for the case $\alpha =2$. The interface in that case has macroscopic, rather than mesoscopic flucuations, which makes our proof break down.

Neither is it clear to us whether the methods of 
Littin and Picco \cite{LitP} will allow  to extend the entropic repulsion results to other $\alpha$ values, although we expect them to hold also in that regime.

We give lower bounds for the entropic repulsion, that is the size of the ``wet'' region  but have neither checked if upper bounds are feasible, nor if the entropic repulsion holds all the way up to the critical point.


{\bf Acknowledgements:} We thank S. Bethuelsen, M. Cassandro, L. Cioletti, D. Con\-ache, R. Fern\'andez, S. Gallo, G. Iacobelli, G. Maillard, F.  Paccaut and E. Verbitskiy for various helpful conversations over the years. We thank Evgeny Verbitskiy for providing us with \cite{BFV} and Jorge Littin for providing us with \cite{LitP}.

RB is partially supported by the Dutch stochastics cluster STAR, by FAPESP Grant 2011/16265-8, and  CNPq grants 453985/2016-5, 312112/2015-7 and 446658/2014-6.
 
EOE is supported by FAPESP  grants 14/10637-9, 15/14434-8.

ALN has benefited from various supports (STAR, CNRS, EURANDOM, Lorentz Center, TU Delft, RU Groningen) for short research  visits to the Netherlands.

\end{document}